\documentclass[envcountsect,envcountsame]{llncs}
\usepackage[hidelinks]{hyperref}
\usepackage{amsmath,amsfonts,paralist,import,mathtools,color}
\usepackage[capitalise,noabbrev,nameinlink]{cleveref}
\usepackage[T1]{fontenc}
\hypersetup{hypertexnames=false} 

\title{Improved Complexity Analysis of Quasi-Polynomial Algorithms Solving Parity Games}

\author{Paweł Parys\orcidID{0000-0001-7247-1408}\thanks{Author supported by the National Science Centre, Poland (grant no.\@ 2021/\allowbreak 41/\allowbreak B/\allowbreak ST6/\allowbreak 03914).} \and
	Aleksander Wiącek}
\institute{Institute of Informatics, University of Warsaw, Poland}

\newcommand{\Gg}{\mathcal{G}}

\newcommand{\Oo}{\mathcal{O}}
\newcommand{\Tt}{\mathcal{T}}
\newcommand{\Nat}{\mathbb{N}}
\newcommand{\Even}{\mathsf{Even}}
\newcommand{\Odd}{\mathsf{Odd}}

\newcommand{\ceil}[1]{\lceil#1\rceil}
\newcommand{\floor}[1]{\lfloor#1\rfloor}

\newcommand{\fk}{f}
\newcommand{\Fk}{F}
\crefname{equation}{Equation}{Equations}

\begin{document}

\maketitle

\begin{abstract}
	We improve the complexity of solving parity games (with priorities in vertices) for $d=\omega(\log n)$ by a factor of $\Theta(d^2)$:
	the best complexity known to date was $\Oo(mdn^{1{.}45+\log_2(d/\log_2 n)})$,
	while we obtain $\Oo(mn^{1{.}45+\log_2(d/\log_2 n)}/d)$,
	where $n$ is the number of vertices, $m$ is the number of edges, and $d$ is the number of priorities.

	We base our work on existing algorithms using universal trees, and we improve their complexity.
	We present two independent improvements.
	First, an improvement by a factor of $\Theta(d)$ comes from a more careful analysis of the width of universal trees.
	Second, we perform (or rather recall) a finer analysis of requirements for a universal tree:
	while for solving games with priorities on edges one needs an $n$-universal tree,
	in the case of games with priorities in vertices it is enough to use an $n/2$-universal tree.
	This way, we allow solving games of size $2n$ in the time needed previously to solve games of size $n$;
	such a change divides the quasi-polynomial complexity again by a factor of $\Theta(d)$.

\end{abstract}

\section{Introduction}

	Parity games have played a fundamental role in automata theory, logic, and their applications to verification and synthesis since early 1990s.
	The algorithmic problem of finding the winner in parity games can be seen as the algorithmic back end to problems in the automated verification and controller synthesis.
	It is polynomial-time equivalent to the emptiness problem for nondeterministic automata on infinite trees with parity acceptance conditions, and to the model-checking problem for modal $\mu$-calculus~\cite{EJS01}.
	It lies also at the heart of algorithmic solutions to the Church’s synthesis problem~\cite{RabinBook}.
	Moreover, decision problems for modal logics like validity or satisfiability of formulae in these logics can be reduced to parity game solving.
	Some ideas coming originally from parity games allowed obtaining new results concerning translations between automata models for $\omega$-words~\cite{parity2AWA,parity2AWA-univ},
	as well as concerning relatively far areas of computer science, like Markov decision processes~\cite{FearnleyMDP} and linear programming~\cite{FHZ-simplex,Friedmann-Zadeh}.

	The complexity of solving parity games, that is, deciding which player has a winning strategy, is a long standing open question.
	The problem is known to be in $\mathsf{UP}\cap\mathsf{coUP}$~\cite{up-co-up}
	(a subclass of $\mathsf{NP}\cap\mathsf{coNP}$)
	and the search variant (i.e., to find a winning strategy) is in \textsf{PLS}, \textsf{PPAD}, and even in their subclass \textsf{CLS}~\cite{Daskalakis-Papadimitriou}.
	The study of algorithms for solving parity games has been dominated for over two decades by algorithms
	whose run-time was exponential in the number of distinct priorities~\cite{Zielonka,BCJLM97,Seidl96,old-progress-measure,strategy-improvement,Schewe-big-steps,priority-promotion},
	or mildly subexponential for a large number of priorities~\cite{randomized-subexponential,subexponential}.
	The breakthrough came in 2017 from Calude, Jain, Khoussainov, Li, and Stephan~\cite{calude}
	who gave the first quasi-polynomial-time algorithm using the novel idea of \emph{play summaries}.
	Several other quasi-polynomial-time algorithms were developed soon after~%
	\cite{progress-measure,ordered-qpt,Lehtinen,Zielonka-Parys,Zielonka-Parys-journal,Strahler-number,symmetric-lifting,Parysian-flair}.

	Fijalkow~\cite{fijalkow} made explicit the concept of \emph{universal trees} that is implicit in the \emph{succinct tree coding} result of Jurdziński and Lazić~\cite{progress-measure}.
	It was then observed that universal trees are not only present, but actually necessary in all existing quasi-polynomial-time approaches for solving parity games~%
	\cite{universal-trees,Zielonka-Parys-universal,fixpoints}.
	Namely, it was shown that any algorithm solving a parity game with $n$ vertices, and following existing approaches,
	needs to operate on an $n$-universal tree (as defined in the sequel).
	There is, however, a catch here: this necessity proof is for parity games with priorities on edges, while quite often one considers less succinct games with priorities in vertices.
	(A general method for switching from priorities on edges to priorities in vertices is to replace each vertex by $d$ vertices, one for each priority,
	and to redirect every edge to a copy of the target vertex having the appropriate priority; then the number of vertices changes from $n$ to $nd$.)

	The main contribution of this paper lies in a finer analysis of the width of $n$-universal trees.
	We improve the upper bound on this width (and thus also the upper bound on the running time of algorithms using such trees) by a factor of $\Theta\left(\frac{d}{\log n}\right)$,
	where $d$ is the number of priorities.

	Then, a second improvement is obtained by recalling from Jurdziński and Lazić~\cite{progress-measure} that
	in order to solve parity games with priorities in vertices it is enough to use an $\floor{n/2}$-universal tree
	instead of an $n$-universal trees, exploiting the ``catch'' mentioned above.
	This allows solving games of size $2n$ in the time needed previously to solve games of size $n$
	and, in consequence, improves the complexity of solving parity games (with priorities in vertices) once again by a factor of $\Theta\left(\frac{d}{\log n}\right)$.

	Combining the two improvements, we decrease the upper bound on the complexity of solving parity games with $d=\omega(\log n)$ by a factor of $\Theta\left(\frac{d^2}{\log^2 n}\right)$:
	the best bound known to date~\cite{fijalkow,progress-measure}, namely $\Oo(mdn^{1{.}45+\log_2(d/\log_2 n)})$, is decreased to the bound $\Oo(mn^{1{.}45+\log_2(d/\log_2 n)}/d)$.
	We remark that both bounds do not display polylogarithmic factors; they are dominated by $n^{\Oo(1)}$,
	where the $\Oo(1)$ comes from the difference between $1{.}45$ and the actual constant $\log_2 e$ which should appear in the exponent
	(the same style of writing the bound is employed in prior work).
	Thus, while writing the bounds in such a form, the improvement is by a factor of $\Theta(d^2)$.
	Simultaneously, the two observations become too weak to improve asymptotics of the complexity in the case of $d=\Oo(\log n)$.

\section{Preliminaries}

\paragraph{Parity games.}

	A \emph{parity game} is a two-player game between players Even and Odd
	played on a \emph{game graph} defined as a tuple $\Gg=(V,V_\Even,E,d,\pi)$,
	where $(V,E)$ is a nonempty finite directed graph in which every vertex has at least one successor;
	its vertices are labelled with positive integer \emph{priorities} by $\pi\colon V\rightarrow\{1,2,\dots,d\}$ (for some \emph{even} number $d\in\mathbb{N}$), and
	divided between vertices $V_\Even$ \emph{belonging to Even} and vertices $V_\Odd=V\setminus V_\Even$ \emph{belonging to Odd}.
	We usually denote $|V|$ by $n$ and $|E|$ by $m$.

	Intuitively, the dynamics of the game are defined as follows.
	The play starts in a designated starting vertex.
	Then, the player to whom the current vertex belongs selects a successor of this vertex, and the game continues there.
	After an infinite time, we check for the maximal priority visited infinitely often; its parity says which player wins.

	Formally, we define the winner of a game using positional (i.e., memoryless) strategies.
	An \emph{Even's positional strategy} is a set $\sigma\subseteq E$ of edges such that
	for every vertex $v$ of Even, in $\sigma$ there is exactly one edge starting in $v$,
	and for every vertex $v$ of Odd, in $\sigma$ there are all edges starting in $v$.
	An \emph{Odd's positional strategy} is defined by swapping the roles of Even and Odd.
	An Even's (Odd's) positional strategy $\sigma$ is winning from a vertex $v$ if for every infinite path in the subgraph $(V,\sigma)$,
	the maximal priority occurring infinitely often on this path is even (odd, respectively).
	We say that Even/Odd \emph{wins} from $v$ if Even/Odd has a positional strategy winning from $v$.

	The above definition accurately describes the winning player due to \emph{positional determinacy} of parity games:
	from every vertex $v$ of a parity game, one of the players, Even or Odd, wins from $v$~\cite{positional-determinacy}.
	It follows that if a player can win from a vertex by a general strategy (not defined here), then he can win also by a positional strategy.

\paragraph{Trees and universal trees.}

	In progress measure algorithms~\cite{progress-measure,fijalkow,Strahler-number} strategies are described by mappings from game graphs to ordered trees, defined as follows:
	an \emph{ordered tree} (or simply a \emph{tree}) $T$ of height $h$ is a finite connected acyclic graph, given together with a linear order $\leq_x$ for every node $x$ thereof, such that
	\begin{compactitem}
	\item	there is exactly one node with in-degree $0$, called a \emph{root}; every other node has in-degree $1$;
	\item	for every \emph{leaf} (i.e., a node of out-degree $0$), the unique path from the root to this leaf consists of $h$ edges;
	\item	$\leq_x$ is a linear order on the \emph{children} of $x$ (i.e., nodes to which there is an edge from $x$),
		which describes the left-to-right ordering of these children.
	\end{compactitem}
	The \emph{width} of a tree $T$ is defined as its number of leaves and denoted $|T|$.

	Let $T_1$ and $T_2$ be two ordered trees of the same height.
	We say that $T_1$ \emph{embeds} into $T_2$ if there is an injective mapping $f$ preserving the child relation and the $\leq_x$ relations:
	\begin{compactitem}
	\item	if $y$ is a child of $x$ in $T_1$, then $f(y)$ is a child of $f(x)$ in $T_2$, and
	\item	if $y\leq_x z$ in $T_1$, then $f(y)\leq_{f(x)} f(z)$ in $T_2$.
	\end{compactitem}
	
	\begin{figure}
		\centering
		\def\svgscale{0.5}\import{pics/}{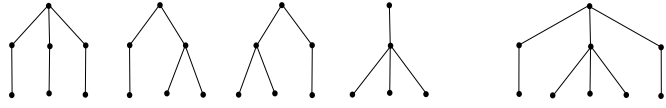}\vspace{-1ex}
		\caption{All four trees of height 2 and width 3 (left); a 3-universal tree of height 2 (right)}
		\label{fig:1}
	\end{figure}

	A tree $\Tt$ of height $h$ is \emph{$n$-universal} if every ordered tree of height $h$ and width at most $n$ embeds into $\Tt$.
	Consult \cref{fig:1} for an example.

\section{On the width of universal trees}

	In this section we prove an improved upper bound on the width of universal trees, as described by the following theorem:

	\begin{theorem}\label{thm:upper-bound}
		For all $n,h\in\Nat$, where $n\geq 1$, there exists an $n$-universal tree of height $h$ and width at most $\fk(n,h)$, where
		\begin{align*}
			\fk(n, h) \leq n\cdot\binom{h - 1 + \floor{\log_2 n}}{\floor{\log_2 n}}
				\leq n^{2{.}45-\varepsilon + \log_2 \left( 1 + \frac{h - 1}{\log_2 n} \right)}
		\end{align*}
		for some $\varepsilon>0$.
		Additionally, $\fk(n,h)=\Oo(n^{2{.}45-\varepsilon+\log_2(h/\log_2 n)})$ if $h=\omega(\log n)$.
	\end{theorem}

	We remark that as the height $h$ we usually take (numbers close to) $d/2$, where $d$ is the number of priorities in a parity game.
	The value $2{.}45-\varepsilon$ in the exponent means that the actual constant is slightly smaller than $2{.}45$.

	Before giving a proof, let us compare the above bound with bounds known to date:
	Fijalkow~\cite[Theorem 4]{fijalkow} derives an upper bound of $2^{\ceil{\log_2 n}} \binom{h - 1 + \ceil{\log_2 n}}{\ceil{\log_2 n}}$.
	The difference is thus in replacing $\ceil{\log_2 n}$ by $\floor{\log_2 n}$.
	For ``a majority'' of $n$ these two values differ by $1$, in which case the quotient of the two binomial coefficients is
	$\frac{h - 1 + \ceil{\log_2 n}}{\ceil{\log_2 n}}$;
	this quotient is in $\Theta\left(\frac{h}{\log n}\right)$ if $h=\omega(\log n)$.
	Coming to the asymptotics, the bound of Fijalkow is in $\Oo(hn^{2{.}45-\varepsilon+\log_2(h/\log_2 n)})$ (see Jurdziński and Lazić~\cite[Lemma 6]{progress-measure} for a proof),
	which is indeed worse by a factor of $\Theta(h)$ than our bound.

	In order to prove \cref{thm:upper-bound}, we use a construction of Fijalkow~\cite[Theorem 4]{fijalkow}
	(which is essentially the same construction as in Jurdziński and Lazić~\cite{progress-measure}).
	He shows that there exists an $n$-universal tree of height $h$ and width $\fk(n, h)$, where the function $\fk$
	(extended to $n=0$ by setting $\fk(0, h)=0$, and to $h=0$, $n\geq 1$ by setting $\fk(n,0)=1$)
	is defined by the following recursive formula:
	\begin{align}
		\fk(0, h) &= 0&& \text{for } h\geq 0,\nonumber\\
		\fk(n, 0) &= 1&& \text{for } n\geq 1,\nonumber\\
		\fk(n, h) &= \fk(n, h - 1) + \fk(\floor{n / 2}, h) + \fk(n - 1 - \floor{n / 2}, h)&& \text{for } n, h \geq 1.
			\label{eq:fk-recursive}
	\end{align}

	Compared to Fijalkow~\cite{fijalkow}, we thus perform a more careful analysis of the above recursive formula.
	First, we provide an explicit formula for $\fk(n, h)$:

	\begin{lemma}\label{lem:fk-explicit}
		The function $\fk$ can be described by the following explicit formula, for all $n,h\geq1$:
		\begin{align*}
			\fk(n, h) = \sum_{i = 0}^{\lfloor \log_2 n \rfloor - 1} 2^i\cdot\binom{h - 1 + i}{h - 1}
				+ (n - 2^{\lfloor \log_2 n \rfloor} + 1)\cdot\binom{h - 1 + \lfloor \log_2 n \rfloor}{h - 1}.
		\end{align*}
	\end{lemma}

	While it is possible to confirm the correctness of Lemma~\ref{lem:fk-explicit} by an inductive proof,
	we present here a proof based on generating functions, as it is more instructive.

	\begin{proof}[\cref{lem:fk-explicit}]
		Let $\Fk(x, y) = \sum_{n, h \geq 1} \fk(n, h) x^n y^h$ be a generating function of the two-dimensional sequence $\fk(n, h)$,
		excluding the values for $n = 0$ and for $h = 0$.
		We multiply both sides of \cref{eq:fk-recursive} by $x^n y^h$,
		and then we sum the result over all $n, h \geq 1$;
		we obtain that
		\begin{align*}
			\hspace{-1em}\Fk(x, y) & = \sum_{\substack{n \geq 1 \\ h \geq 1}} \fk(n, h - 1) x^n y^h + \sum_{\substack{k \geq 1 \\ h \geq 1}} (\fk(k, h) + \fk(k - 1, h)) x^{2 k} y^h\\
				&\hspace{18em} + \sum_{\substack{k \geq 0 \\ h \geq 1}} 2 \fk(k, h) x^{2 k + 1} y^h.
		\end{align*}
		Next, in the first sum above we shift $h$ by $1$, and we move components with $\fk(n,0)$ into a separate sum;
		in the second sum above we split $\fk(k, h) + \fk(k - 1, h)$ into two separate sums, where in the latter we shift $k$ by $1$,
		and we move components with $\fk(0,h)$ into a separate sum;
		in the last sum above, we move components for $k=0$ into a separate sum.
		We obtain
		\begin{align*}
			\hspace{-1em}\Fk(x, y) & = \sum_{n \geq 1} \fk(n, 0) x^n y + \sum_{\substack{n \geq 1 \\ h \geq 1}} \fk(n, h) x^n y^{h + 1}\\
				& \quad + \sum_{\substack{k \geq 1 \\ h \geq 1}} \fk(k, h) (x^2)^k y^h + \sum_{h \geq 1} \fk(0, h) x^2 y^h
					+ \sum_{\substack{k \geq 1 \\ h \geq 1}} \fk(k, h) (x^2)^k y^h x^2\\
				& \quad + \sum_{h \geq 1} 2 \fk(0, h) x y^h + \sum_{\substack{k \geq 1 \\ h \geq 1}} 2 \fk(k, h) (x^2)^k y^h x\\
				& = \frac{x y}{1 - x} + y \Fk(x, y) + \Fk(x^2, y) + 0 + x^2 \Fk(x^2, y) + 0 + 2 x \Fk(x^2, y).
		\end{align*}
		This gives us the following equation concerning the generating function:
		\begin{align}\label{F-func-equation}
			\Fk(x, y) = y \Fk(x, y) + \Fk(x^2, y) (1 + 2 x + x^2) + \frac{x y}{1 - x}.
		\end{align}

		Let $H(x, y) = \frac{1 - x}{y} \Fk(x, y)$; then $\Fk(x, y) = \frac{y}{1 - x} H(x, y)$.
		Note that $H$ is the generating function representing the differences between values of $\fk$ for successive values of $n$, with $h$ shifted by $1$.
		Substituting this to \cref{F-func-equation} we obtain:
		\begin{align*}
			\frac{y}{1 - x} H(x, y) &= \frac{y^2}{1 - x} H(x, y) + \frac{y}{1 - x^2} H(x^2, y) (1 + x)^2 + \frac{x y}{1 - x},\\
			\frac{y}{1 - x} H(x, y) &= \frac{y^2}{1 - x} H(x, y) + \frac{y}{1 - x} H(x^2, y) (1 + x) + \frac{x y}{1 - x},\\
			H(x, y) &= y H(x, y) + H(x^2, y) (1 + x) + x.
		\end{align*}
		Next, let us write $H(x, y) = \sum_{n \geq 1} x^n h_n(y)$.
		We substitute this to the equation above:
		\begin{align*}
			\sum_{n \geq 1} x^n h_n(y) = \sum_{n \geq 1} x^n y h_n(y) + \sum_{n \geq 1} (x^{2 n} + x^{2 n + 1}) h_n(y) + x.
		\end{align*}
		In order to find $h_n$, we compare coefficients in front of $x^n$, on both sides of the equation.
		For $n = 1$ we have
		\begin{align*}
			h_1(y) = y h_1(y) + 1,&&\mbox{so}&&h_1(y) = \frac{1}{1 - y},
		\end{align*}
		and for $n \geq 2$ we have
		\begin{align*}
			h_n(y) = y h_n(y) + h_{\lfloor n / 2 \rfloor}(y),&&\mbox{so}&&h_n(y) = \frac{h_{\lfloor n / 2 \rfloor}(y)}{1 - y}.
		\end{align*}
		It follows that for all $n\geq 1$ we have a formula
		\begin{align*}
			h_n(y) = \frac{1}{(1 - y)^{\lfloor \log_2 n \rfloor + 1}}.
		\end{align*}

		Below, we use the notation $[x^n]A(x)$ for the coefficient in front of $x^n$ in the function $A(x)$.
		We also use the following formula, where $k\geq 1$:
		\begin{align}\label{eq:geom-power}
			\frac{1}{(1 - x)^k} = \sum_{n = 0}^\infty \binom{n + k - 1}{n} x^n.
		\end{align}

		We now conclude the proof (here we assume that $n,h\geq 1$):
		\begin{align*}
			\fk(&n, h)  = [x^n y^h] \Fk(x, y) = [x^n y^h] \frac{y}{1 - x} H(x, y)\\
			& = [y^{h - 1}] \sum_{j = 0}^n [x^j] H(x, y) = [y^{h - 1}] \sum_{j = 1}^n h_j(y)\\
			& = [y^{h - 1}] \left( \sum_{i = 0}^{\lfloor \log_2 n \rfloor - 1} 2^i h_{2^i}(y) + (n - 2^{\lfloor \log_2 n \rfloor} + 1)  h_{2^{\lfloor \log_2 n \rfloor}}(y) \right)\\
			& = \sum_{i = 0}^{\lfloor \log_2 n \rfloor - 1} 2^i [y^{h - 1}] \frac{1}{(1 - y)^{i + 1}}
				+ (n - 2^{\lfloor \log_2 n \rfloor} + 1) [y^{h - 1}] \frac{1}{(1 - y)^{\lfloor \log_2 n \rfloor + 1}}\\
			& = \sum_{i = 0}^{\lfloor \log_2 n \rfloor - 1} 2^i\cdot\binom{h - 1 + i}{h - 1} + (n - 2^{\lfloor \log_2 n \rfloor} + 1)\cdot\binom{h - 1 + \lfloor \log_2 n \rfloor}{h - 1}.
		\end{align*}
		Above, the third line is obtained based on the fact that $h_{2^i} = h_{2^i + 1} = \ldots = h_{2^{i + 1} - 1}$,
		and the last line is obtained based on \cref{eq:geom-power}.
		This finishes the proof of \cref{lem:fk-explicit}.
	\end{proof}

	Next, we give two auxiliary lemmata, useful while bounding the asymptotics:

	\begin{lemma}\label{lem:ineq1}
		For all $x\geq 0$ it holds that $\ln(1 + x) \cdot (1 + x) \geq x$.
	\end{lemma}

	\begin{proof}
		Denote $h(x) = \ln(1 + x) \cdot (1 + x) - x$.
		Because $h(0) = 0$, it is enough to prove that the function $h$ is increasing.
		To this end, we compute its derivative.
		We obtain $h'(x) = \ln(1 + x)$, so $h'(x) > 0$ for $x > 0$ and indeed the function $h$ is increasing.
	\end{proof}

	\begin{lemma}\label{lem:ineq2}
		For every $c\geq 0$ the function
		\begin{align*}
			\alpha_c(x) = \left(1 + \frac{c}{x} \right)^x
		\end{align*}
		is nondecreasing for $x > 0$.
	\end{lemma}

	\begin{proof}
		We compute the derivative of $\alpha_c$:
		\begin{align*}
			\alpha'_c(x) = \alpha_c(x)\cdot \left( \ln \left( 1 + \frac{c}{x} \right) - \frac{c}{x\cdot \left(1 + \frac{c}{x} \right)} \right).
		\end{align*}
		Because $\alpha_c(x) \geq 0$ for $x > 0$, in order to confirm that $\alpha'_c(x) \geq 0$ it is enough to check that
		\begin{align*}
			\ln \left( 1 + \frac{c}{x} \right) \geq \frac{c}{x\cdot \left(1 + \frac{c}{x} \right)}.
		\end{align*}
		To show this inequality, we multiply to both its sides $1 + \frac{c}{x}$ and we use \cref{lem:ineq1}.
	\end{proof}

	We are now ready to finish the proof of \cref{thm:upper-bound}:

	\begin{proof}[\cref{thm:upper-bound}]
		In order to obtain $\fk(n, h) \leq n\cdot \binom{h - 1 + \floor{\log_2 n}}{\floor{\log_2 n}}$,
		we replace all binomial coefficients in the formula of \cref{lem:fk-explicit} by $\binom{h - 1 + \floor{\log_2 n}}{\floor{\log_2 n}}$;
		obviously $\binom{h - 1 + i}{h - 1} \leq \binom{h - 1 + \floor{\log_2 n}}{h - 1}=\binom{h - 1 + \floor{\log_2 n}}{\floor{\log_2 n}}$ for $i \leq \floor{\log_2 n}$.

		Recall that $x^{\log_2 y}=y^{\log_2 x}$ for any $x,y>0$.
		The second inequality from the theorem's statement is obtained using, consecutively,
		the estimation $\binom{n}{k} \leq \left( \frac{e n}{k} \right)^k$, the inequality $\log_2 e < 1{.}45$,
		and \cref{lem:ineq2}:
		\begin{align*}
			\fk(n,h)&\leq n\cdot\binom{h - 1 + \lfloor \log_2 n \rfloor}{\lfloor \log_2 n \rfloor}
				\leq n\cdot\left( e \left( 1 + \frac{h - 1}{\lfloor \log_2 n \rfloor} \right) \right)^{\lfloor \log_2 n \rfloor} \\
				&\leq n\cdot e^{\log_2 n}\cdot\left( 1 + \frac{h - 1}{\lfloor \log_2 n \rfloor} \right)^{\lfloor \log_2 n \rfloor}
				= n^{1 + \log_2 e}\cdot\left( 1 + \frac{h - 1}{\lfloor \log_2 n \rfloor} \right)^{\lfloor \log_2 n \rfloor} \\
				&\leq n^{2{.}45-\varepsilon}\cdot \left( 1 + \frac{h - 1}{\log_2 n} \right)^{\log_2 n}
				= n^{2{.}45-\varepsilon + \log_2 \left( 1 + \frac{h - 1}{\log_2 n} \right)}.
		\end{align*}

		Assume now that $h=\omega(\log n)$.
		Then
		\begin{align*}
			\log_2\left(1+\frac{h-1}{\log_2 n}\right)&=\log_2\left(\frac{h}{\log_2n}\cdot(1+o(1))\right)\\
			&=\log_2\left(\frac{h}{\log_2n}\right)+\log_2(1+o(1)).
		\end{align*}
		The component $\log_2(1+o(1))=o(1)$ can be removed at the cost of decreasing the constant $\varepsilon$, hence we obtain that
		\begin{align*}
			\fk(n,h)\leq n^{2{.}45-\varepsilon + \log_2 \left( 1 + \frac{h - 1}{\log_2 n} \right)}
				=\Oo\left(n^{2{.}45-\varepsilon + \log_2 \left(\frac{h}{\log_2 n} \right)}\right).
		\tag*{\qed}\end{align*}
	\end{proof}

\section{Using smaller trees}

	Papers showing how to solve parity games with use of universal trees~\cite{fijalkow,Strahler-number,Zielonka-Parys-universal,symmetric-lifting}
	assume that in order to solve games with $n$ vertices one needs $n$-universal trees.
	And for games with priorities on edges it can be shown that $n$-universality is indeed required~\cite{universal-trees,fixpoints}.
	However most papers (including ours) define parity games with priorities on vertices, in which case the necessity proof does not apply.

	Let us come back to the paper of Jurdziński and Lazić~\cite{progress-measure}.
	Although it does not mention universal trees explicitly, it uses a very particular universal tree,
	called succinct tree coding.
	The important point is that this tree is not $n$-universal, but rather $\eta$-universal, where $\eta$ can be either the number of vertices of odd priority,
	or the number of vertices of even priority, whatever is smaller, so clearly $\eta\leq\floor{n/2}$
	(formally, throughout their paper $\eta$ denotes the number of vertices of even priority,
	but they explain at the beginning of Section 4 that priorities can be shifted by $1$ ensuring that $\eta\leq\floor{n/2}$).
	Moreover, by looking into their paper one can see that the only ``interface'' between Section 2, defining the succinct tree coding,
	and later sections, using the coding, is Lemma 1, where it is shown that the succinct tree coding is an $\eta$-universal tree.
	It is then easy to see that the algorithm of Jurdziński and Lazić works equally well with any other $\eta$-universal tree (of appropriate height, namely $d/2$)
	in place of the succinct tree coding.
	In particular, we can use the universal trees from our \cref{thm:upper-bound}, being of smaller width.

	Let us now bound the width of a universal tree that is needed for the algorithm, comparing it with previous approaches.
	In order to avoid the additional parameter $\eta$, we replace it in the sequel by $\floor{n/2}$, making use of the inequality $\eta\leq\floor{n/2}$.
	Substituting $d/2$ for $h$ and $\floor{n/2}$ for $n$ in the formula from \cref{thm:upper-bound}, we obtain that the width of an $\floor{n/2}$-universal tree $\Tt$ of height $d/2$ can satisfy
	\begin{align*}
		|\Tt| \leq \left\lfloor\frac{n}{2}\right\rfloor\cdot\binom{d/2 - 1 + \floor{\log_2 \floor{n/2}}}{\floor{\log_2 \floor{n/2}}}.
	\end{align*}
	To compare, for an $n$-universal tree $\Tt'$ by \cref{thm:upper-bound} we have
	\begin{align*}
		|\Tt'| \leq n\cdot\binom{d/2 - 1 + \floor{\log_2 n}}{\floor{\log_2 n}}.
	\end{align*}
	For natural $n$ we always have $\floor{\log_2 \floor{n/2}}=\floor{\log_2(n/2)}=\floor{\log_2 n}-1$,
	so the quotient of the two binomial coefficients is $\frac{d/2 - 1 + \floor{\log_2 n}}{\floor{\log_2 n}}$,
	which is in $\Theta\left(\frac{d}{\log n}\right)$ if $d=\omega(\log n)$.

	Let us now determine the asymptotics for $d=\omega(\log n)$.
	First, let us simplify the formula for the $n$-universal tree $\Tt'$:
	\begin{align*}
		|\Tt'|=\Oo(n^{2{.}45-\varepsilon+\log_2(d/2/\log_2 n)})
			=\Oo(n^{1{.}45-\varepsilon+\log_2(d/\log_2 n)}).
	\end{align*}
	For the $\floor{n/2}$-universal tree $\Tt$ we thus have
	\begin{align}\label{eq:4}
		|\Tt|=\Oo\left(\left\lfloor\frac{n}{2}\right\rfloor^{1{.}45-\varepsilon+\log_2(d/\log_2 \floor{n/2})}\right)
			=\Oo\left(\left(\frac{n}{2}\right)^{1{.}45-\varepsilon+\log_2(d/\log_2 n)}\right).
	\end{align}
	The second equality above is obtained by replacing $\left\lfloor\frac{n}{2}\right\rfloor$ with the slightly greater value of $\frac{n}{2}$, and by observing that
	\begin{align*}
		\log_2\left(\frac{d}{\log_2\floor{n/2}}\right)&=\log_2\left(\frac{d}{\log_2n}\cdot(1+o(1))\right)\\
		&=\log_2\left(\frac{d}{\log_2n}\right)+\log_2(1+o(1));
	\end{align*}
	the component $\log_2(1+o(1))=o(1)$ can be removed at the cost of decreasing the constant $\varepsilon$.
	We continue by analysing the logarithm of the bound:
	\begin{align}
		&\log_2\left(\frac{n}{2}\right)\cdot\left(1{.}45-\varepsilon+\log_2\frac{d}{\log_2 n}\right)
		=\left(\log_2n-1\right)\cdot\left(1{.}45-\varepsilon+\log_2\frac{d}{\log_2 n}\right)\nonumber\\
		&\hspace{5em}\leq\log_2n\cdot\left(1{.}45-\varepsilon+\log_2\frac{d}{\log_2 n}\right)-\log_2\frac{d}{\log_2 n}+\varepsilon\nonumber\\
		&\hspace{5em}\leq\log_2n\cdot\left(1{.}45-\varepsilon+o(1)+\log_2\frac{d}{\log_2 n}\right)-\log_2d+\varepsilon;\label{eq:5}
	\end{align}
	the last equality above was obtained by observing that
	\begin{align*}
		-\log_2\frac{d}{\log_2 n}=-\log_2d+\log_2\log_2n=-\log_2d+\log_2n\cdot o(1).
	\end{align*}
	Combining \cref{eq:4,eq:5}, we obtain the following bound on $|\Tt|$:
	\begin{align*}
		|\Tt|=\Oo\left(n^{1{.}45-\varepsilon+o(1)+\log_2(d/\log_2 n)}\cdot\frac{1}{d}\right).
	\end{align*}
	The $o(1)$ component can be removed at the cost of decreasing the constant $\varepsilon$ again;
	we can thus write
	\begin{align*}
		|\Tt|=\Oo\left(n^{1{.}45-\varepsilon+\log_2(d/\log_2 n)}\cdot\frac{1}{d}\right).
	\end{align*}

	We now come to the complexity of the algorithm itself.
	As observed by Jurdziński and Lazić~\cite[Theorem 7]{progress-measure}, their algorithm, when using a universal tree $\widehat\Tt$,
	works in time $\Oo(m\cdot\log n\cdot\log d\cdot|\widehat\Tt|)$.
	Using $\Tt$ as $\widehat\Tt$, and observing that the polylogarithmic function $O(\log n\cdot\log d)$ can be again ``eaten'' by the $-\varepsilon$ component of the exponent,
	we obtain the following bound on the complexity, which is our main theorem:

	\begin{theorem}
		For $d=\omega(\log n)$ one can find a winner in a parity game with $n$ vertices, $m$ edges, and $d$ priorities in time
		\begin{align*}
			\Oo\left(m\cdot n^{1{.}45+\log_2(d/\log_2 n)}\cdot\frac{1}{d}\right).
		\end{align*}
	\end{theorem}

	\begin{remark}
		Compared to the previous bound of $\Oo\left(m\cdot d\cdot n^{1{.}45+\log_2(d/\log_2 n)}\right)$ (from Fijalkow~\cite{fijalkow}),
		we obtain an improvement by a factor of $\Theta(d^2)$.
		One $\Theta(d)$ is gained in \cref{thm:upper-bound}, by improving the bound for the size of a universal tree.
		A second $\Theta(d)$ is gained by using $\floor{n/2}$-universal trees instead of $n$-universal trees.
	\end{remark}

	\begin{remark}
		The complexity obtained in Fijalkow~\cite{fijalkow} is the same as in Jurdziński and Lazić~\cite{progress-measure}:
		he gains a factor of $\Theta(d)$ by using a better construction of a universal tree instead of the ``succinct tree coding'',
		but simultaneously he loses a factor of $\Theta(d)$ due to using $n$-universal trees in place of $\floor{n/2}$-universal trees.

		Let us also use this place to note that there is a small mistake in the paper of Jurdziński and Lazić~\cite{progress-measure}.
		Namely, in the proof of Lemma 6 they switch to analysing a simpler expression $\binom{\ceil{\log_2\eta}+d/2}{\ceil{\log_2\eta}}$
		in place of $\binom{\ceil{\log_2\eta}+d/2+1}{\ceil{\log_2\eta}+1}$, saying that the simpler expression is within a constant factor of the latter one.
		This statement is false, though: the quotient of the two expressions is
		$\frac{\ceil{\log_2\eta}+d/2+1}{\ceil{\log_2\eta}+1}$, which is in $\Theta\left(\frac{d}{\log \eta}\right)$ if $d=\omega(\log\eta)$.
		Thus, the actual complexity upper bound resulting from their analysis of the algorithm for $d=\omega(\log\eta)$ should not be
		$\Oo\left(m\cdot d\cdot\eta^{1{.}45+\log_2(d/\log_2\eta)}\right)$ as they claim,
		but rather $\Oo\left(m\cdot d^2\cdot\eta^{1{.}45+\log_2(d/\log_2\eta)}\right)$
		(and taking $n/2$ for $\eta$ this gives us the complexity $\Oo\left(m\cdot d\cdot n^{1{.}45+\log_2(d/\log_2 n)}\right)$, the same as in Fijalkow~\cite{fijalkow}).
	\end{remark}
	
	\begin{remark}
		While the current paper was under preparation,
		Dell'Erba and Schewe~\cite{Dell_Erba_2022} published some other improvement of Jurdziński and Lazić's algorithm~\cite{progress-measure}.
		A complete discussion is difficult to perform, since the authors have not provided a precise bound.
	\end{remark}

\bibliography{bib}

\end{document}